\documentclass{scrartcl}

\usepackage[utf8]{inputenc}
\usepackage[T1]{fontenc}
\usepackage[UKenglish]{babel}
\usepackage{mathtools}
\usepackage{amsthm}
\usepackage{tikz-cd}
\usepackage{quiver}
\usepackage{url}
\usepackage{hyperref}
\usepackage[disable]{todonotes} 
\usepackage{comment}

\usepackage[
, textwidth = 140mm
, textheight = 222mm
, margin = 38mm
, headheight = 3mm
, headsep = 10mm
, footskip = 5mm
]{geometry}

\usepackage{tabu}
\newcolumntype{H}{>{\setbox0=\hbox\bgroup}c<{\egroup}}

\usepackage{general-macros}
\usepackage{local-macros}

\usepackage{cleveref}
\usepackage{thmstyle}

\usepackage[normalem]{ulem}

\usepackage{tikz-cd}
\usepackage{float}
\usetikzlibrary{arrows, automata, shapes, calc, fadings} 
\tikzset{->, auto, >=stealth', font=\small}
\tikzset{state/.style={shape=circle, draw, fill=white, initial text=,
    inner sep=.5mm, minimum size=1.5mm}}
\tikzset{accepting/.style=accepting by arrow}
\tikzset{state with output/.style={shape=rectangle split, rectangle
    split parts=2, draw, fill=white,
    initial text=, inner sep=1mm}}
\tikzset{
  pomset/.style={
    commutative diagrams/every diagram,
    row sep=0.02em,
    column sep=0.2cm,
    commutative diagrams/cramped,
    every matrix/.append style={inner sep=+-0.3em},
    every cell/.append style={inner sep=+0.3em},
    baseline=-1pt,
    every matrix/.append style={
      matrix of math nodes,
      left delimiter=(,
      right delimiter=),
      commutative diagrams/cramped
    }
  }
}

\definecolor{formalcolour}{HTML}{1E88E5}

\title{Irrationality of Process Replication for Higher-Dimensional Automata}
\date{Pre-Print}
\author{%
  Thomas Baronner\thanks{Student, Leiden University \url{mailto:thomasbaronner@gmail.com}}
  \and
  Henning Basold\thanks{LIACS, Leiden University, \url{mailto:h.basold@liacs.leidenuniv.nl}}
  \and
  Márton Hablicsek\thanks{MI, Leiden University \url{mailto:m.hablicsek@math.leidenuniv.nl}}
}

\begin{document}

\maketitle

\begin{abstract}
Higher-dimensional automata (HDA) are a formalism to faithfully model the behaviour of
concurrent systems.
For ordinary automata, there is a correspondence between regular expressions, regular languages and
finite automata, which provides a powerful link between algebraic proofs and operational behaviour.
It has been shown by Fahrenberg~et~al. that finite HDA correspond with interfaced interval pomset
languages generated by sequential and parallel composition and non-empty iteration, and thereby to
a variant of Kleene algebras (KA) with parallel composition.
It is known that this correspondence cannot be extended to concurrent KA, which
additionally have process replication.
An alternative to finite HDA are locally finite HDA, in which every state can only reach finitely many other states, and finitely branching HDA.
In this paper, we show that both classes of HDA are closed under process replication and thus models of concurrent KA.
To achieve this, we prove that the category of HDA is locally finitely presentable, where the finite HDA generate all other HDA.
We then prove that this has the unfortunate side-effect that all HDA are locally finite, which means that the correspondence with concurrent KA trivialises.
Similarly, we also show that, even though finitely branching HDA are closed under process replication, the resulting HDA necessarily have infinitely many initial states.
\end{abstract}

\section{Introduction}
\label{sec:introduction}

Automata theory has as a core goal that problems, like deciding language membership, should be
solved by finitary means.
With this goal in mind, research on automata typically strives for a correspondence between
certain kinds of finitary automata, languages, syntactic expressions, and algebras.
The classical example of this correspondence is between finite (non-)deterministic automata,
regular languages, free Kleene algebras (regular expressions), and finite syntactic monoids.
In the area of concurrency, such correspondences have been sought as
well~\cite{EN04:HigherDimensionalAutomata,%
FJS+22:KleeneTheoremHigherDimensional,%
Grabowski81:PartialLanguages,%
KBL+19:SeriesparallelPomsetLanguages,%
LW00:SeriesParallelLanguages}.
Several automata models have emerged from this as did the notion of concurrent Kleene
algebras~\cite{HMS+09:ConcurrentKleeneAlgebra,HMS+11:ConcurrentKleeneAlgebra}, which extend
Kleene algebras with parallel computation and process replication (also called parallel closure).
Concurrent Kleene algebras (CKA) correspond to several
automata models~\cite{KBL+19:SeriesparallelPomsetLanguages,LW00:SeriesParallelLanguages}.

Parallel to automata models for CKA, several operational models of
true concurrency have been developed, such as Petri nets and higher-dimensional automata (HDA).
These are models that can faithfully represent parallel computation without having
to resort to sequentialisation~\cite{Glabbeek06:ExpressivenessHDA}.
HDA have received a lot of attention because of the
geometric view on concurrency that they offer~\cite{Fahrenberg13:HistoryPreservingBisimilarityHigherDimensional,%
Fajstrup98:DetectingDeadlocksConcurrent,%
Goubault00:GeometryConcurrencyUser,%
Kahl18:LabeledHomologyHigherdimensional,%
Pratt91:ModelingConcurrencyGeometry,%
Raussen21:ConnectivitySpacesDirected,%
Glabbeek06:ExpressivenessHDA}.
Fahrenberg~et~al. proved a correspondence between finite HDA and Kleene algebras (KA)
with parallel composition and that KA with process replication cannot be given semantics in
terms of finite HDA~\cite[Lemma~12]{FJS+22:KleeneTheoremHigherDimensional}.
We show in this paper that process replication can also not be realised as neither locally compact HDA, in which
every state can only reach finitely many other states~\cite{BMS13:SoundCompleteAxiomatizations,%
Milius10:SoundCompleteCalculus,%
Milius13:RationalOperationalModels}, nor as finitely branching HDA with finitely initial states.
Our approach is to prove that the category of HDA is locally finitely presentable, which allows us to
define the language of HDA in terms of languages of finite HDA~\cite{Fahrenberg21:LanguagesHigherDimensionalAutomata},
prove that any HDA is locally compact HDA and that process replication cannot be realised in any finitary way over HDA.

Let us briefly discussion the intuition behind HDA.
The idea is that they generalise labelled transition systems to allow for $n$ actions to be
active simultaneously by modelling transitions as $n$-cells in higher-dimensional cubes.
For instance, \cref{fig:transitions-hda} shows a graphical representation of a HDA over an alphabet
with actions $\set{a,b,c,d}$.
\begin{figure}[t]
  \centering
  \vspace*{-1.5em}
  \begin{tikzcd}[row sep=0.6cm, column sep=2cm,
    execute at end picture={
      \foreach \n in  {A,B,...,F}
      {\coordinate (\n) at (\n.center);}
      \fill[formalcolour!40,opacity=0.3]
      (A) -- (B) -- (C) -- (D) -- cycle
      (D) -- (C) -- (E) -- (F) -- cycle;
    }]
    |[inner sep=0pt]|{\cdot} \ar[r, "b"]
    & |[alias=G,inner sep=0pt]|{\cdot}
    & |[alias=H,inner sep=0pt]|{\cdot}
    \\
    |[alias=B,inner sep=0pt]|{\cdot} \ar[r, "b"] \ar[u, "c"]
    & |[alias=C,inner sep=0pt]|{\cdot} \ar[r, "d"] \ar[u, "c"]
    & |[alias=E,inner sep=0pt]|{\cdot} \ar[u, "c"]
    \\
    |[alias=A,inner sep=0pt]|{\cdot} \ar[r, "b"{swap}] \ar[u, "a"]
    \ar[urr,dashrightarrow,controls={+(0.5,0.9) and +(-1,-0.5)}]
    \ar[urr,dashrightarrow,controls={+(2.5,0.3) and +(-1,-0.5)}]
    & |[alias=D,inner sep=0pt]|{\cdot} \ar[r, "d"{swap}] \ar[u, "a"]
    & |[alias=F,inner sep=0pt]|{\cdot} \ar[u, "a"]
    \ar[Rightarrow, from=G, to=G, start anchor={east}, end anchor={[xshift=2em]east}]
    \ar[Rightarrow, from=H, to=H, start anchor={east}, end anchor={[xshift=2em]east}]
  \end{tikzcd}
  \vspace*{-0.5em}
  \caption[Transitions in Higher-Dimensional Automata]{
    Event $a$ may happen in parallel with $b$ and $d$ (filled squares), while $c$ is
    in conflict with $b$ and $d$ (not filled);
    two parallel executions of $a$ and $b$, and $a$ and $d$ are indicated by the dashed homotopic
    paths; cells with double arrows are accepting cells}
\label{fig:transitions-hda}
\end{figure}
The dots indicate $0$-cells, in which no action is active, solid arrows are $1$-cells that are
transitions with one active action, and the blue shaded areas are $2$-cells with two active actions.
Starting from the bottom left, first $a$ and $b$ may be active in parallel and any execution
path through the shaded area is allowed.
In the square above that, the action $c$ and $b$ have to be executed sequentially because the
square is not filled.
The HDA accepts a run if one of the $0$-cells with a double arrow
is reached.
For instance, the (sequential) path $a \to b \to c$ is accepted.
HDA accept in general pomset languages~\cite{Fahrenberg21:LanguagesHigherDimensionalAutomata}.
In the case of \cref{fig:transitions-hda}, the accepted language is the following
set consisting of ten pomsets.
\begin{equation*}
  \biggl\{
    \left(
  \begin{tikzcd}[row sep=0.05em, column sep=0.3cm, cramped]
    a \ar[r] & b \ar[r] & c
  \end{tikzcd}
  \right),
    \left(
  \begin{tikzcd}[row sep=0.05em, column sep=0.3cm, cramped]
    a \ar[r] & c \ar[r] & b
  \end{tikzcd}
  \right),
    \left(
  \begin{tikzcd}[row sep=0.05em, column sep=0.3cm, cramped]
    b \ar[r] & a \ar[r] & c
  \end{tikzcd}
  \right),
  \left(
  \begin{tikzcd}[row sep=0.05em, column sep=0.3cm, cramped]
    a \ar[r] & b \ar[r] & d \ar[r] & c
  \end{tikzcd}
  \right),
    \left(
  \begin{tikzcd}[row sep=0.05em, column sep=0.3cm, cramped]
    b \ar[r] & d \ar[r] & a \ar[r] & c
  \end{tikzcd}
  \right)
\end{equation*}
\vspace*{-1em}
\begin{equation*}
  \begin{tikzpicture}[pomset]
    \matrix[name=m, commutative diagrams/cramped] {
      a &   \\
        & c \\
      b &   \\
    };
    \path[commutative diagrams/.cd, every arrow, every label]
    (m-1-1) edge (m-2-2)
    (m-3-1) edge (m-2-2);
  \end{tikzpicture},
  \begin{tikzpicture}[pomset]
    \matrix[name=m, commutative diagrams/cramped] {
        & a &   \\
        &   & c \\
      b & d &   \\
    };
    \path[commutative diagrams/.cd, every arrow, every label]
    (m-1-2) edge (m-2-3)
    (m-3-1) edge (m-3-2)
    (m-3-2) edge (m-2-3);
  \end{tikzpicture},
  \begin{tikzpicture}[pomset]
    \matrix[name=m, commutative diagrams/cramped] {
      a &   &   \\
        & d & c \\
      b &   &   \\
    };
    \path[commutative diagrams/.cd, every arrow, every label]
    (m-1-1) edge (m-2-2)
    (m-2-2) edge (m-2-3)
    (m-3-1) edge (m-2-2);
  \end{tikzpicture},
  \begin{tikzpicture}[pomset]
    \matrix[name=m, commutative diagrams/cramped] {
        & a &   \\
      b &   & c \\
        & d &   \\
    };
    \path[commutative diagrams/.cd, every arrow, every label]
    (m-1-2) edge (m-2-3)
    (m-2-1) edge (m-1-2) edge (m-3-2)
    (m-3-2) edge (m-2-3);
  \end{tikzpicture}
  \Biggr\}
\end{equation*}
The first six are purely sequential runs, while the last four run $a$, $b$, $c$ and $d$ in parallel.
Pomset languages can be composed with the operations of concurrent Kleene algebras, and
one may then ask which of these operations carry over to HDA and may result in a correspondence
between (locally) finite HDA and rational pomset languages constructed from these operations.

\paragraph*{Outline and Contributions}
\label{sec:contributions}

We show in \cref{sec:compact-hda} that the category of HDA is
locally finitely presentable (lfp) and that finite HDA are exactly the compact (or finitely presentable) objects.
This allows the reduction of arguments to finite HDA.
In \cref{sec:comp-hda}, we show that languages of coproducts and filtered colimits
of HDA are given directly by the languages of the HDA in the corresponding diagrams, and that
this fails for general colimits.
We also give in \cref{sec:hda-monoidal} a novel characterisation of the tensor product of
HDA, and then use this and the lfp property to show that the tensor product yields the
parallel composition of languages.
In \cref{sec:replication} we present two possible local finiteness conditions for HDA that are stable under process replication.
We then show that both notions involve some infinite branching and we end with a proof
that it is impossible to realise process replication without infinite branching.
We begin with a recap of the theory of pomset languages
in \cref{sec:conc-words} and of HDA in \cref{sec:hda}.

\paragraph*{Related Work}
\label{sec:related-work}

The work of Lodaya and Weil~\cite{LW00:SeriesParallelLanguages} offers another automaton model for
concurrency, called branching automata, as well as an algebraic perspective.
Interestingly, their correspondence is restricted to languages of bounded width.
Our result in \cref{sec:replication} could be extended to show that finitely branching HDA
correspond to languages of bounded width, but we do not explore this further, as bounded width
languages can be realised without process replication.

Ésik and Németh~\cite{EN04:HigherDimensionalAutomata} prove a correspondence between rational
languages of \emph{series-parallel biposets}, which are essentially pomsets, and finite
parenthesising automata.
Such automata have two kinds of states and transition relations that can be thought of as
$0$- and $1$-cells, and transitions among them (respectively $1$- and $2$-cells) and transitions
up and down one dimension and that are guarded by parentheses.
Thus, they make HDA more flexible in that they allow dimension change but also restrict the
dimensions.

Jipsen and Moshier~\cite{JM16:ConcurrentKleeneAlgebra} reiterate on branching
automata~\cite{LW00:SeriesParallelLanguages} but improve them by adding a bracketing condition
akin to parenthesising automata~\cite{EN04:HigherDimensionalAutomata}.

Kappé et al.~\cite{Kappe20:ConcurrentKleeneAlgebra,%
KBL+17:BrzozowskiGoesConcurrent,%
KBL+19:SeriesparallelPomsetLanguages}
have shown that finite well-nested pomset automata correspond to concurrent Kleene algebras and,
what they call, series-parallel rational expressions.
Pomset automata have two transition functions, one for sequential and one for parallel computation.
The latter can branch out to finitely many parallel states and synchronise after each has completed
their work.
This allows them to implement process replication because the number of parallel processes can grow
arbitrarily during execution, while the dimension of a cell in a HDA fixes the number of parallel
processes.
We will discuss this in \cref{sec:concl}.

Finally, our work builds on that of Fahrenberg~et~al.~\cite{FJS+22:KleeneTheoremHigherDimensional}.
For the most part, we follow them in our definitions of HDA and languages, but also deviate in some choices, like the definition of the cube
category and the tensor product of HDA.
We have also incorporated their insight to give up event consistency~\cite{Fahrenberg21:LanguagesHigherDimensionalAutomata},
as the category of HDA would otherwise not be
cocomplete~\cite{Baronner22:FiniteAccessibilityHigherDimensional}.

\paragraph*{Acknowledgements}
\label{acknowledgements}

We would like to thank the referees for their valuable comments, in particular the suggestion of an alternative proof
strategy for \cref{hda-lfp}.

\section{Concurrent Words via Ipomsets}
\label{sec:conc-words}
In this section, we recap the theory of interval ipomsets and their languages,
sequential composition, parallel composition and
parallel Kleene closure~\cite{Fahrenberg21:LanguagesHigherDimensionalAutomata}.

\subsection{Ipomsets}
\label{sec:ipomsets}

\begin{definition} \label{labelled-iposet:def}
  A \emph{labelled iposet} $P$ is a tuple
  $\parens*{\car{P},<_{P},\dr_{P},S_{P},T_{P},\lambda_{P}}$ where
\begin{itemize}
    \item $\car{P}$ is a finite set,
    \item $<_{P}$ is a strict partial order on $\car{P}$ called \emph{precedence order},
    \item $\dr_{P}$ is a strict partial order on $\car{P}$, called \emph{event order}, that is
      linear on $<_{P}$-antichains,
    \item $\lambda_{P} \from \car{P} \to \Sigma$ is a labelling map to an alphabet $\Sigma$,
    \item $S_{P} \subseteq \car{P}$ is a set of $<_{P}$-minimal elements called the
      \emph{source set}, and
    \item $T_{P} \subseteq \car{P}$ is a set of $<_{P}$-maximal elements called the
      \emph{target set}.
\end{itemize}

\end{definition}
We write $\emptyLO$ for the empty iposet.
Note that the condition that $\dr_{P}$ is linear on $<_{P}$-antichains implies that the union $\dr_{P} \cup <_{P}$ is a total order.
\begin{definition}
  We say that a labelled iposet $P$ is \emph{subsumed} by a labelled iposet $Q$,
  written $P\sqsubseteq Q$, if there exists a bijection $f \from \car{P} \to \car{Q}$ with
  $f\parens*{S_P} = S_Q$, $f \parens*{T_P} = T_Q$ and such that for all $x,y\in \car{P}$ we have
  \begin{enumerate}
  \item $f(x)<_Qf(y)\implies x<_Py$
  \item $x\dr_Py\text{, }x\not<_Py\text{, }y\not<_Px\implies f(x)\dr_Qf(y)$
  \item $\lambda_P(x)=\lambda_Q\circ f(x)$
  \end{enumerate}
  The labelled iposets $P$ and $Q$ are isomorphic if $f$ is an isomorphism for both orders.
  An \emph{ipomset} is an isomorphism class of labelled iposets.
\end{definition}
$P\sqsubseteq Q$ intuitively means that $P$ is more ordered by the precedence order $<$ than $Q$,
which means that $P$ has less ``concurrency''.
Isomorphisms between labelled iposets are unique, which means that any skeleton of the
category of labelled iposets and subsumptions is isomorphic to the quotient by isomorphisms.
\begin{definition}
  An ipomset $P$ is an \emph{interval ipomset} if there is a pair of functions
  $b,e \from \car{P} \to \R$ to the real numbers, such that $b(x) \leq e(x)$ for all
  $x\in \car{P}$ and $x <_P y \iff e(x)<b(y)$ for all $x,y\in \car{P}$.
  The pair of functions $(b,e)$ is called an \emph{interval representation} of $P$.
  We let $\iiPom$ be the set of all interval ipomsets.
\end{definition}

The simplest example of an ipomset that is not interval is the ipomset $P$ with $\car{P}=\left\{a,b,c,d\right\}$ with only $a<_{P}b$, $c<_{P}d$, $a \dr_{P} c$ and $b \dr_{P}d$. This is the ipomset variant of the $(2+2)$-poset.
Note that the event order plays no role in the interval representation.

Given a set of interval ipomsets $A\subseteq\iiPom$, the down-closure of $A$ is defined as
usual by $\downCl{A} = \setDef{P\in\iiPom}{\exist{Q \in A} P\sqsubseteq Q}$.
\begin{definition} \label{lang-definition}
  A \emph{language} $L$ of interval ipomsets is a down-closed set of interval ipomsets, that is,
  if $\downCl{L} \subseteq L$ holds.
  We denote by $\Lang$ the thin category with languages as objects and subset inclusions as
  morphisms.
\end{definition}

\subsection{Composition of ipomsets and languages}
\label{sec:composition}

\begin{definition}
  \label{def:glueing}
  We say that ipomsets $P$ and $Q$ \emph{sequentially match} if there is a (necessarily unique)
  isomorphism $f \from (T_P, \dr_{P}) \to (S_Q, \dr_{Q})$ with $\lambda_{Q} \comp f = \lambda_{P}$.
  If $P$ and $Q$ match sequentially, then we define the \emph{gluing composition} by
  \begin{equation*}
    P \glue Q
    = \parens*{\car{P \glue Q},<_{P \glue Q},\dr_{P \glue Q},S_P,T_Q,\lambda_{P \glue Q}} \, ,
  \end{equation*}
  where $(\car{P \glue Q}, \dr_{P \glue Q})$ is the pushout
  $\colim\parens*{(\car{P}, \dr_{P}) \hookleftarrow T_{P} \xrightarrow{f} (\car{Q}, \dr_{Q})}$
  of posets of $f$ along the inclusion.
  The precedence order $<_{P \glue Q}$ is the union of the images of
  $<_P$, $<_Q$ and $(\car{P} \setminus T_{P}) \times (\car{Q} \setminus S_{Q})$ in
  $\car{P \glue Q}$.
  Finally, the labelling function $\lambda_{P \glue Q} \from \car{P \glue Q} \to \Sigma$ is defined
  as the copairing $\copair{\lambda_P, \lambda_{Q}}$ on the pushout using that $f$ preserves
  labelling.
\end{definition}

If $P$ and $Q$ are interval ipomsets, then their gluing composition $P\glue Q$ is an interval
ipomset as well~\cite[Lem.~41]{Fahrenberg21:LanguagesHigherDimensionalAutomata}.
This uses that the map $f$, which attaches the interfaces, is an order isomorphism
and that the event order is linear.

If the interfaces $T_{P}$ and $S_{Q}$ are empty, then $P \glue Q$ is the coproduct of
$(\car{P}, \dr_{P})$ and $(\car{Q}, \dr_{Q})$, and at the same time the join of
$(\car{P}, <_{P})$ and $(\car{Q}, <_{Q})$ considered as categories.
This amounts to the serial pomset composition~\cite{FJS+22:KleeneTheoremHigherDimensional}, which is the
generalisation of concatenation of words to pomsets.
\begin{definition}\label{def:glueing-lang}
  The sequential composition of languages $L_1$ and $L_2$ is defined as
\begin{equation*}
  L_1\glue L_2=\downCl{\setDef{P\glue Q}{
      P\in L_1, Q\in L_2 , \text{and } P \text{ and } Q \text{ match sequentially}}}
\end{equation*}
\end{definition}
\begin{definition} \label{def:parallel}
  We define the \emph{parallel composition} of ipomsets $P$ and $Q$ by
\[
  P\parallel Q = \left(\car{P}+\car{Q},<_{P \parallel Q},\dr_{P \parallel Q},S_{P\parallel Q},T_{P\parallel Q},\lambda_{P\parallel Q}\right)
\]
Let $i_P:\car{P}\to\car{P}+\car{Q}$ and $i_Q:\car{Q}\to\car{P}+\car{Q}$ be the canonical injection maps. Using these injection maps we define $<_{P\parallel Q}=i_P\left(<_P\right)\cup i_Q\left(<_Q\right)$, $S_{P\parallel Q}=i_P\left(S_P\right)\cup i_Q\left(S_Q\right)$, $T_{P\parallel Q}=i_P\left(T_P\right)\cup i_Q\left(T_Q\right)$ and
$\lambda_{P\parallel Q} = \copair{\lambda_{P},\lambda_{P}}$.
Then $\dr_{P\parallel Q}$ is defined as the ordered sum of the event orders, in other words, $i_P$ preserves the order $\dr_P$ as $\dr_{P\parallel Q}$ and $i_Q$ preserves $\dr_Q$ as $\dr_{P\parallel Q}$ and for all $x\in\car{P}$, $y\in\car{Q}$ we have $i_P(x)\dr_{P\parallel Q}i_Q(y)$.
\end{definition}

Differently said, the event order $\dr_{P\parallel Q}$ on the parallel composition $P\parallel Q$ is defined as the join of $(\car{P},\dr_{P})$ and $(\car{Q},\dr_{Q})$ thought of as categories.

\begin{definition} \label{def:paralell-lang}
  The parallel composition of languages $L_1$ and $L_2$ is defined as
\[
L_1\parallel L_2=\downCl{\left\{P\parallel Q\mid P\in L_1\text{, }Q\in L_2\right\}}
\]
and the \emph{parallel Kleene closure} of a language $L$ as
\begin{equation*}
  L^{(\ast)} = \bigcup_{n \in \N} L^{\parallel n}
  \quad \text{where} \quad
  L^{\parallel 0} = \set{\emptyLO}
  \, \text{ and } \,
  L^{\parallel (n+1)} = L \parallel \bparens{L^{\parallel n}} \, .
\end{equation*}
\end{definition}

Down-closure is needed in Definitions \ref{def:glueing-lang} and \ref{def:paralell-lang}, since sequential or parallel compositions of down-closed languages may not result in a down-closed language.
However, we can form unions of languages.
\begin{lemma} \label{down-closure-union}
  Languages are closed under arbitrary unions.
\end{lemma}
\begin{proof}
Let $L \from D\to\Lang$ be a small diagram of down-closed interval ipomset languages and let $L_{\cup}=\bigcup_{d\in D}L_d$ be their union. Then for every $Q\in L_{\cup}$ there exists at least one $d\in D$ such that $Q\in L_d$, which means that $Q$ has to be an interval ipomset. Moreover for every $P\in\iiPom$ with $P\sqsubseteq Q$ we by definition have $P\in L_d$ which means that we have to have $P\in L_{\cup}$ as well. Therefore $L_{\cup}$ is down-closed as well.
\end{proof}

We conclude this section by showing that the parallel composition of languages respects small colimits, which are unions in the category of languages.

\begin{lemma} \label{parallel-lang-union-commute}
  For small diagrams $M:D\to\Lang$ and $N:E\to\Lang$ of languages we have
  \begin{equation*}
    \smashoperator{\bigcup_{(d,e)\in D\times E}} M_d\parallel N_e
    = \parens*{\bigcup\nolimits_{d\in D}M_d} \bigparallel[\Big] \parens*{\bigcup\nolimits_{e\in E}N_e}
  \end{equation*}
\end{lemma}
\begin{proof}
Suppose that $L_1=\bigcup_{(d,e)\in D\times E}M_d\parallel N_e$ and $L_2=\left(\bigcup_{d\in D}M_d\right)\parallel\left(\bigcup_{e\in E}N_e\right)$.

Suppose that $R\in L_1$. Then there exist $d\in D$ and $e\in E$ such that $R\in M_d\parallel N_e$. Then there exists a $P\in M_d$ and a $Q\in N_e$ such that $R\sqsubseteq P\parallel Q$. Since $P\in\bigcup_{d\in D}M_d$ and $Q\in\bigcup_{e\in E}N_e$ this means that $P\parallel Q\in L_2$ and therefore $R\in L_2$. This gives us $L_1\subseteq L_2$.

Suppose that $R\in L_2$. Then there exists a $P\in \bigcup_{d\in D}M_d$ and a $Q\in \bigcup_{e\in E}N_e$ such that $R\sqsubseteq P\parallel Q$. Therefore there exist $d\in D$ and $e\in E$ such that $P\in M_d$ and $Q\in N_e$, which means that $P\parallel Q\in M_d\parallel N_e$ and therefore $P\parallel Q\in L_1$. This gives us $R\in L_1$ and therefore $L_1\supseteq L_2$ which means that we have $L_1=L_2$.
\end{proof}

\section{Higher-Dimensional Automata}
\label{sec:hda}

In this section we first recall the definition of HDA, then discuss the monoidal structure
of HDA to model parallel computation and finally show in \cref{sec:compact-hda} that
the category of HDA is locally finitely presented by finite HDA.

\subsection{The Category of HDA}
\label{sec:hda-cat}

Higher-dimensional automata are modelled as labelled precubical sets, which in turn are presheaves
over a category of basic hypercubes.
Such cubes can be represented as ordered sets, where the size of the set corresponds to the
dimension of the cube, and the morphism of the ordered sets determine how the faces of $(n+1)$-cells
in a precubical set match with $n$-dimensional faces.
We fix from now on an alphabet $\Sigma$ in which HDA are labelled.

\begin{definition}
  A labelled linearly ordered set or lo-set $\left(U,\dr,\lambda\right)$ is a finite set $U$ with a
  strict linear order $\dr$ and a labelling map $\lambda \from U \to \Sigma$.
  We write $\emptyLO$ for the unique empty lo-set.
  A lo-map is a map between lo-sets that preserves the order and the labelling.
  Lo-sets and -maps form a category $\lSLO$.
\end{definition}

The category $\lSLO$ is monoidal with $U \cubeTens V$ being the join of $U$ and $V$ considered as thin
categories and the monoidal unit being the empty set.
Explicitly, the underlying set of $U \cubeTens V$ is the coproduct $U + V$, the order is given by
$x \dr_{U \cubeTens V} y$ iff $x \dr_{U} y$, $x \dr_{V} y$, or $x \in U$ and $y \in V$.
The labelling $\lambda_{U \cubeTens V}$ is given by the copairing
$\copair{\lambda_{U}, \lambda_{V}} \from U + V \to \Sigma$.

Note that lo-maps are necessarily injective, which means that morphisms $f \from U \to V$ in
$\lSLO$ are equivalently defined by their image $f(U)$ or their complement $V \setminus f(U)$.
Moreover, $f$ is an isomorphism iff $f$ is surjective, i.e. if $V \setminus f(U) = \emptyset$.
Since isomorphisms in $\lSLO$ are unique, we can safely identify it with a skeleton that
has as objects pairs $(\fOrd{n}, w)$ where $n \in \N$, $\fOrd{n}$ is the finite ordinal
$\set{0 < \dotsm < n - 1}$ with $n$ elements and $w \in \Sigma^n$ is a word of length $n$.

\begin{definition} \label{coface-def}
  A \emph{coface map} $d \from U \to V$ between lo-sets $U$ and $V$ is a triple $(f, A, B)$,
  where $f \from U \to V$ is a lo-map and $\set{A,B}$ is a partition of the complement image of $f$,
  that is, $V \setminus f(U) = A \cup B$ and $A \cap B = \emptyset$.
  We write $d(x)$ for the application of the underlying map $f$ to $x$ to simplify notation.
  For $A,B \subseteq U$ that are disjoint, we denote by $d_{A,B} \from U \setminus (A \cup B) \to U$
  the coface map $(i, A, B)$, where $i \from U \setminus (A \cup B) \to U$ is the inclusion.
\end{definition}

The monoidal structure on $\lSLO$ carries over to a monoidal structure on the category of lo-sets and coface
maps, which is the \emph{full precube category} $\FCube$.

\begin{lemma}
  \label{lem:precube-category}
  The lo-sets and coface maps form a monoidal category $(\FCube, \cubeTens, \cubeUnit)$.
\end{lemma}
\begin{proof}
  Composition of $(e, C, D) \from V \to W$ and $(d, A, B) \from U \to V$ is given by
  $(e, C, D) \comp (d, A, B) = (e \comp d, e(A) \cup C, e(B) \cup D)$.
  That $\set{e(A) \cup C, e(B) \cup D}$ form a partition of the complement image of $e \comp d$
  follows from injectivity of $e$, properties of the image and the given partitions.
  The identity is given by $(\id, \emptyset, \emptyset)$, and the unit and associativity axioms
  follow from colimit preservation of the image.
  The monoidal structure in inherited from $\lSLO$: on objects we use $\sloTens$ and on morphisms
  we take
  $(d_{1}, A_{1}, B_{1}) \cubeTens (d_{2}, A_{2}, B_{2}) = (d_{1} \sloTens d_{2}, A_{1} \sloTens A_{2}, B_{1} \sloTens B_{2})$,
  where we write $A_{1} \sloTens A_{2}$ for the application of $\sloTens$ to the inclusions
  $A_{k} \subseteq V$.
  Finally, the associator and unitor isomorphisms have empty complement image that can be
  trivially partitioned.
\end{proof}

Since isomorphisms in $\lSLO$ are unique, they are in $\FCube$ as well and we can use the same
skeleton as for $\lSLO$ only with the morphisms of $\FCube$.
We denote this small skeleton by $\Cube$.

\begin{definition}
  A \emph{precubical set} is a presheaf (functor) $X \from \op{\square} \to \SetC$ and a \emph{precubical map} is a natural transformation.
  They form a category $\presheaf{\square}$.
  We write $\Yo$ for the Yoneda embedding $\Cube \to \presheaf{\Cube}$ with
  $\Yo_{U} = \Cube(-, U)$.
\end{definition}

We refer to the elements of $X[U]$ as \emph{cells} and to the cardinality of $U$ as
the \emph{dimension} of those cells.
If for some $U$ of cardinality $n$ the set $X[U]$ is inhabited and for all $V$ with
cardinality greater than $n$ the sets $X[U]$ are empty, then we say that $X$ has finite
dimension $n$.
A precubical set $X$ is finite if it has finite dimension and if for all $U \in \Cube$ the
set $X[U]$ is finite.

To ease notation, we denote
the face map $X[d_{A,B}] \from X[U] \to X\brak*{U \setminus \parens*{A\cup B}}$, induced by a coface map
$d_{A,B} \from U \setminus \parens*{A\cup B} \to U$, by $\delta_{A,B}$.
The face maps $\delta_{A,\emptyset}$ and $\delta_{\emptyset,B}$ will be suggestively abbreviated
to $\delta^0_A$ and $\delta^1_B$.

\begin{definition}
  \label{def:hda}
  A \emph{higher-dimensional automaton (HDA)} is a tuple $\HDATup{X}$ where $X$ is a
  precubical set, $\sCells{X}$ and $\aCells{X}$ are families of sets indexed by lo-sets
  of starting and accepting cells with $(\sCells{X})_{U} \subseteq X_{U}$
  and $(\aCells{X})_{U} \subseteq X_{U}$.
  A HDA map $f \from \HDATup{X} \to \HDATup{Y}$ is a precubical map $f \from X \to Y$ that
  preserves the starting and accepting cells, that is,
  $f\parens*{\sCells{X}} \subseteq \sCells{Y}$ and $f\parens*{\aCells{X}}\subseteq \aCells{Y}$.
  We denote by $\HDA$ the category of higher-dimensional automata and their maps.
\end{definition}

We usually leave out the index on $\sCells{X}$ and $\aCells{X}$ for better readability.

\begin{lemma}
  \label{lem:hda-precubical-adjunction}
  The forgetful functor $\forget:\HDA \to \presheaf{\Cube}$ has left and right
  adjoints $N$ and $T$ given, respectively, by
  $NX = (X, \emptyset, \emptyset)$ and $TX = (X, \obj{X}, \obj{X})$
  where $\obj{X}$ is the family obtained by forgetting the action of $X$ on morphisms.
  Thus, the left adjoint $N$ stipulates no starting or accepting cells, while $T$ considers all
  cells as starting and accepting.
\end{lemma}

\subsection{Monoidal Structure on HDA}
\label{sec:hda-monoidal}

Our main interest in this paper is to realise (repeated) parallel composition of languages as HDA.
In this section we briefly discuss how HDA can be synchronised in parallel via a monoidal
product on $\HDA$.
\begin{definition}
  \label{def:hda-tensor}
  The tensor product of HDA is the
  Day convolution~\cite{Day70:ConstructionBiclosedCategories,IK86:UniversalPropertyConvolution,%
    Loregian20:CoendCalculus}, which is given
  for HDA $X$ and $Y$ on the precubical sets by the following coend.
  \begin{equation*}
    X \otimes Y = \int^{V,W} \Cube(-, V \cubeTens W) \times X[V] \times Y[W]
  \end{equation*}
  The starting cells $\left(X\otimes Y\right)_{\bot}$ are given as the image
  of all inclusions
  \begin{equation*}
    \begin{tikzcd}[column sep=small]
      (X_{\bot} \cap X[V]) \times (Y_{\bot} \cap Y[W])
      \ar[r]
      & \Cube(V \cubeTens W,V \cubeTens W) \times X[V] \times Y[W]
      \ar[r]
      & X \otimes Y
    \end{tikzcd}
  \end{equation*}
  and analogously for accepting cells $\left(X\otimes Y\right)^{\top}$.
  A diagram chase shows that $\otimes$ is well-defined on HDA morphisms.
  For any $U \in \Cube$, we can make $\Yo_{U}$ an HDA by taking all cells to be starting and accepting.
  The monoidal unit $I$ is the Yoneda embedding $\Yo_{\varepsilon}$ of the empty lo-set with the
  only $0$-cell being starting and accepting.
\end{definition}
By this definition, the Yoneda embedding is a strong monoidal
functor and $\otimes$ preserves colimits~\cite{IK86:UniversalPropertyConvolution}.
Moreover, $\forget$ is clearly a strict monoidal functor.
Usually, the tensor product of (pre)cubical sets is defined as a
coproduct~\cite{BH87:TensorProductsHomotopies,FJS+22:KleeneTheoremHigherDimensional,%
  Grandis09:DirectedAlgebraicTopology,Kan55:AbstractHomotopy}.
Indeed, we present in appendix~\ref{sec:convolution} a proof for the isomorphism $(X \otimes Y)(U) \cong \coprod_{U = V \cubeTens W} X[V] \times Y[W]$.

\subsection{Filtered Colimits and Compact HDA}
\label{sec:compact-hda}

Compact (or finitely presentable) objects in a category can be thought of as the analogue of finite sets, relative
to what morphisms in that category perceive as finite.
For instance, compact objects in the category $\VecC_{\R}$ of $\R$-vector spaces are vector spaces with
finite dimension.
In $\SetC$ and $\VecC_{\R}$, arguments can be reduced to arguments about compact objects because
\emph{all} objects in those categories are given as nice colimits of a set of chosen compact
objects.
For instance, each set $U$ is given as a colimit of finite sets by taking the union of the finite subsets of $U$.
This process is given by so-called filtered colimits.
The advantage of breaking down objects to filtered colimits of compact objects is that constructions
on objects can be carried out on a set of compact objects instead.
Categories that admit these kind of reduction are called locally finitely presentable (lfp).

In what follows, we recall the definition of lfp categories~\cite{adamek_rosicky_1994,Riehl16:CatTheoryinContext}, show that the category of
HDA is lfp and that the compact objects are precisely the finite HDA.
A category $\Cat{C}$ is called \emph{essentially small} if it is equivalent to a small category.
We call a category $D$ \emph{filtered} if any finite diagram in $D$
has a cocone, or equivalently if (1) $D$ is inhabited, (2) for any two objects $c,d \in D$ there exists an object $e\in D$ and two 
morphisms $c \rightarrow e \leftarrow d$, and (3) for any two morphisms $f,g \from c \to d$ there exist an object $e\in D$ and a morphism $h \from d \to e$ with $h \comp f = h \comp g$.
A \emph{filtered colimit} in a category $\Cat{C}$ is a colimit of a diagram
$F \from D \to \Cat{C}$ where $D$ is filtered.
We say that an object $X \in \Cat{C}$ is \emph{compact} if the hom-functor
$\Cat{C}(X,-) \from \Cat{C} \to \SetC$ preserves filtered colimits.
Finally, the category $\Cat{C}$ is called \emph{locally finitely presentable (lfp)} if it is
cocomplete, the subcategory $\Cat{C}_{\mathrm{c}}$ of compact objects is essentially small,
and every object in $\Cat{C}$ is isomorphic to a filtered colimit of compact objects.
Many calculations are simplified by the fact that the category $\Cat{C}_{\mathrm{c}}$ is closed
under finite colimits~\cite[Prop.~1.3]{adamek_rosicky_1994}.
One of the important examples of a lfp category is the functor category of precubical sets
$\presheaf{\square}$~\cite[Ex.~1.12]{adamek_rosicky_1994} and that the hom-functor $\Yo_{U}$ is
compact in $\presheaf{\square}$ for all $U \in \Cube$.

In what follows, we show that $\HDA$ is equivalent to a reflective subcategory $\Cat{H}$ of a
presheaf category that is closed under filtered colimits.
This implies that $\Cat{H}$ and $\HDA$ are lfp~\cite[Sec. 1C]{adamek_rosicky_1994}.

We define a category $\lSLOc$ with objects
\begin{equation*}
  \obj{\lSLOc} = \obj{\lSLO} \cup \set{\top, \bot} \times \obj{\lSLO}
\end{equation*}
and morphisms between objects are given as follows.
\begin{equation*}
  \lSLOc(X, Y) =
  \begin{cases}
    \lSLO(X, Y), & X, Y \in \obj{\lSLO} \\
    \set{\ast_{\top}}, & X \in \obj{\lSLO}, Y = (\top, X) \\
    \set{\ast_{\bot}}, & X \in \obj{\lSLO}, Y = (\bot, X) \\
    \emptyset, & \text{otherwise}
  \end{cases}
\end{equation*}
We will write $U_{\top}$ and $U_{\bot}$ instead of $(\top, U)$ and $(\bot, U)$ for $U \in \lSLO$.
Let $\Cat{H}$ be the full subcategory of $\presheaf{\lSLOc}$ of presheaves $P$ for which $P(\ast_{\top})$ and $P(\ast_{\bot})$ are injective.
The idea is that $P(U^{\top})$ and $P(U_{\bot})$ contain the starting and accepting cells of dimension $U$.
We now have to show that $\Cat{H}$ is a reflective subcategory of $\presheaf{\lSLOc}$, closed under filtered colimits and equivalent to $\HDA$.

\begin{lemma}
  \label{lem:h-cat-reflective}
  $\Cat{H}$ is a reflective subcategory of $\presheaf{\lSLOc}$.
\end{lemma}
\begin{proof}
Let $I \from \Cat{H} \to \presheaf{\lSLOc}$ be the inclusion functor.
We construct a left-adjoint $T$ to $I$ using the orthogonal epi-mono factorisation system $(E, M)$  on $\SetC$ as follows.
For a presheaf $P \from \op{(\lSLOc)} \to \SetC$ and $U, V \in \lSLO$, we define a presheaf $TP$ by
$(TP)(U) = P(U)$, $(TP)(k) = Pk$ for $k \from V \to U$ and $(TP)(U_{\top})$ and $(TP)(U_{\bot})$ by the following factorisations into a surjection followed by an injection.
\begin{align*}
  & P(U_{\bot}) \xrightarrow{\eta_{P,U_{\bot}}} (TP)(U_{\bot}) \xrightarrow{(TP)(\ast_{\bot})} P(U)
  \\
  & P(U_{\top}) \xrightarrow{\eta_{P,U_{\top}}} (TP)(U_{\top}) \xrightarrow{(TP)(\ast_{\top})} P(U)
\end{align*}
Since the epi-mono factorisation system is functorial, this assignment defines a functor
$T \from \presheaf{\lSLOc} \to \Cat{H}$.
If we put $\eta_{P,U} = \id_{P(U)}$ for $U \in \lSLO$, then this yields together with the above factorisation a natural transformation $\eta \from \Id \to IT$.
Let now $f \from P \to K$ be map natural transformation with $K \in \Cat{H}$.
Since $(E,M)$ is an orthogonal factorisation system, there is for all $U$ a unique map $\bar{f}_{U_{\bot}}$ filling the following diagram.
\begin{equation*}
\begin{tikzcd}[column sep=large]
	{P(U_{\bot})} & {(TP)(U_{\bot})} & {P(U)} \\
	{K(U_{\bot})} & {K(U_{\bot})} & {K(U)}
	\arrow["{\eta_{P, U_{\bot}}}", two heads, from=1-1, to=1-2]
	\arrow["{T(\ast_{\bot})}", tail, from=1-2, to=1-3]
	\arrow["{f_{U_{\bot}}}"', from=1-1, to=2-1]
	\arrow["\id"', from=2-1, to=2-2]
	\arrow["{K(\ast_{\bot})}"', from=2-2, to=2-3]
	\arrow["{f_{U}}", from=1-3, to=2-3]
	\arrow["{\bar{f}_{U_{\bot}}}"', dashed, from=1-2, to=2-2]
\end{tikzcd}
\end{equation*}
Similarly, there is a unique map $\bar{f}_{U_{\top}} \from (TP)(U_{\top}) \to K(U_{\top})$ with $\bar{f}_{U_{\top}} \comp \eta_{P, U_{\top}} = f_{U_{\top}}$.
If we put $\bar{f}_{U} = f_{U}$, then we obtain a unique natural transformation $\bar{f} \from TP \to K$ with $\bar{f} \comp \eta_{P} = f$.
Thus, $(TP, \eta)$ is a reflection of $P$ along $I$ and thus $T \dashv I$.
The inclusion $I$ is by definition full and thus $\Cat{H}$ is a reflective subcategory of the presheaf category $\presheaf{\lSLOc}$.
\end{proof}

\begin{lemma}
  \label{lem:h-cat-filtered-colim}
  $\Cat{H}$ is closed under filtered colimits.
\end{lemma}
\begin{proof}
Let $\Cat{C}$ be filtered and $D \from \Cat{C} \to \Cat{H}$ a diagram.
Since $\presheaf{\lSLOc}$ is a presheaf category, the colimit $\colim ID$ is computed point-wise.
Thus, it remains to prove that $(\colim ID)(\ast_{\top})$ and $(\colim ID)(\ast_{\bot})$ are injective, where we only prove the first and the second is analogous.
We note that the following diagram is a pullback for all $c \in \Cat{C}$ and $U \in \lSLO$ because $D_c(\ast_{\top})$ is a monomorphism (injective).
\begin{equation*}
\begin{tikzcd}[column sep=large]
	{D_c(U_{\top})} & {D_c(U_{\top})} \\
	{D_c(U_{\top})} & {D_c(U)}
	\arrow["{D_c(\ast_{\top})}"', from=2-1, to=2-2]
	\arrow["{D_c(\ast_{\top})}", from=1-2, to=2-2]
	\arrow["\id"', from=1-1, to=2-1]
	\arrow["\id", from=1-1, to=1-2]
\end{tikzcd}
\end{equation*}
Since filtered colimits commute with finite limits and because $\colim$ preserves identities, the following is also a pullback.
\begin{equation*}
\begin{tikzcd}[column sep=large]
	{(\colim D)(U_{\top})} & {(\colim D)(U_{\top})} \\
	{(\colim D)(U_{\top})} & {D_c(U)}
	\arrow["{(\colim D)(\ast_{\top})}"', from=2-1, to=2-2]
	\arrow["{(\colim D)(\ast_{\top})}", from=1-2, to=2-2]
	\arrow["\id"', from=1-1, to=2-1]
	\arrow["\id", from=1-1, to=1-2]
	\arrow["\lrcorner"{anchor=center, pos=0.125}, draw=none, from=1-1, to=2-2]
\end{tikzcd}
\end{equation*}
Therefore, $(\colim D)(\ast_{\top})$ is a monomorphism and $\colim D \in \Cat{H}$.
\end{proof}

\begin{lemma}
  \label{lem:hda-is-presheaf-cat}
  $\Cat{H}$ is equivalent to $\HDA$.
\end{lemma}
\begin{proof}
  This is obvious by mapping a presheaf $P \in \Cat{H}$ to the HDA $(X, X_{\bot}, X_{\top})$ with $X(U) = P(U)$, $X_{\bot} = \bigcup_{U} P(\ast_{\bot})(U_{\bot})$ and $X_{\top} = \bigcup_{U} P(\ast_{\top})(U_{\top})$.
  This mapping induces clearly a fully faithful functor that is essentially surjective, and is thus part of an equivalence.
\end{proof}

Combining \cref{lem:h-cat-reflective,lem:h-cat-filtered-colim,lem:hda-is-presheaf-cat} gives us immediately that $\HDA$ is lfp~\cite[Sec. 1C]{adamek_rosicky_1994}.
\begin{corollary}
  \label{hda-lfp}
  The category of HDA is locally finitely presentable.
\end{corollary}

We use in what follows sometimes the following more specific description of finite presentations in $\HDA$.
Let $I \from \HDAcomp \to \HDA$ be the inclusion functor of the full subcategory
of compact HDA in $\HDA$.
For a HDA $X$, we denote by $\Comma{I}{X}$ the comma category that has as objects
morphisms $Y \to X$ from a compact HDA $Y$ into $X$, and morphisms are the evident commutative
triangles.
The comma category $\Comma{I}{X}$ is essentially small and closed under finite colimits, thus it is a
filtered category.
We write $U_{X} \from \Comma{I}{X} \to \HDAcomp$ for the domain projection functor.
\begin{corollary} \label{pcs-filtered-colimit-finite-pcs}
  Every HDA $X$ can be canonically expressed as the filtered colimit of finite HDA, that is,
  we have $X \cong \colim U_{X}$.
\end{corollary}

For instance, it follows that compact HDA are the expected finite HDA.
\begin{theorem}
  \label{thm:hda-compact-iff-finite}
  A HDA is compact if and only if it is finite.
\end{theorem}
\begin{proof}

Let $X$ be a compact HDA and let $F:D\to\HDA$ be a filtered diagram of finite HDA with the filtered colimit $(X,\phi)$ as per \cref{pcs-filtered-colimit-finite-pcs}. Then, since $X$ is compact, we have
\[
\colim_{d\in D}\text{Hom}\left(X,F(d)\right)\cong\text{Hom}\left(X,\colim_{d\in D}F(d)\right)\cong\text{Hom}\left(X,X\right)
\]
As a consequence, we get that the identity map $\text{id}_X$ factors through a map $X\rightarrow F(d)$. Since $F(d)$ is a finite HDA, $X$ has to be finite as well.
\end{proof}

\section{Languages of Higher-Dimensional Automata}
\label{sec:hda-lang}
Computations as modelled by HDA can be expressed as higher-dimensional paths running through the HDA from a starting cell to an accepting cell. Each of these accepting paths corresponds to an interval ipomset, which allows us to define the languages of HDA as the set of interval ipomsets it accepts.
We expand here on previous work~\cite{FJS+22:KleeneTheoremHigherDimensional} by removing the restriction to finite HDA and by
showing that HDA languages preserve coproducts and filtered colimits.

\subsection{Paths and languages}
\label{sec:paths}

Let us start by defining paths and their labelling.
\begin{definition} \label{def-paths}
A path (of length $n$) in a precubical set or HDA $X$ is a list
\[
\alpha=\left(x_0,\varphi_1,x_1,\varphi_2,...,\varphi_n,x_n\right)
\]
where $x_k\in X\left[U_k\right]$ are cells for $U_k \in \Box$ and for all $1\leq k\leq n$ we have an
\begin{itemize}
    \item up-step: $\varphi_k=d_{A,\emptyset}\in\Cube\left(U_{k-1},U_k\right)$, $x_{k-1}=\delta^0_A\left(x_k\right)$ and $A=U_k\backslash U_{k-1}$, or
    \item down-step: $\varphi_k=d_{\emptyset,B}\in\Cube\left(U_k,U_{k-1}\right)$, $\delta^1_B\left(x_{k-1}\right)=x_k$ and $B=U_{k-1}\backslash U_k$.
\end{itemize}
\end{definition}
The elements $x_k$ are cells, while the $\varphi_k$ express how these cells are connected. Since for a path we cannot have $\delta^0_A\left(x_{k-1}\right)=x_k$ or $x_{k-1}=\delta^1_B\left(x_k\right)$ it can only move along the direction of the arrows. Two paths where the first ends at the starting cell of the other can be composed as follows.
\begin{definition}
Let $\alpha=\left(x_0,\varphi_1,x_1,...,\varphi_n,x_n\right)$ and $\beta=\left(y_0,\psi_1,y_1,...,\psi_m,y_m\right)$ be two paths in a precubical set or HDA $X$ with $x_n=y_0$. Then we define their concatenation $\alpha*\beta$ as the following path in $X$.
\[
\alpha*\beta=\left(x_0,\varphi_1,x_1,...,\varphi_n,x_n,\psi_1,y_1,...,\psi_m,y_m\right)
\]
\end{definition}
Every path $\alpha=\left(x_0,\varphi_1,x_1,...,\varphi_n,x_n\right)$ can therefore be broken down into paths of length 1, called steps. We denote a step $\left(x_{k-1},\varphi_k,x_k\right)$ with $x_{k-1}\nearrow^Ax_k$ if $\varphi_k=d_{A,\emptyset}$ (an \textit{up step}) or with $x_{k-1}\searrow_Bx_k$ if $\varphi_k=d_{\emptyset,B}$ (a \textit{down step}). We get the unique representation $\left(x_0,\varphi_1,x_1\right)*\left(x_1,\varphi_2,x_2\right)*...*\left(x_{n-1},\varphi_n,x_n\right)$ for the path $\alpha$. Using this we define the labelling of paths recursively.
\begin{definition}
Let $X$ be a precubical set or HDA. Let $\alpha$ be a path in $X$, let $U$ and $V$ be objects in $\Cube$ and let $x\in X[U]$, $y\in X[V]$. Then the labelling $\ev\left(\alpha\right)$ of $\alpha$ is the ipomset that is computed as follows:
\begin{itemize}
\item If $\alpha=\left(x\right)$ is a path of length 0 then its label is
  $\ev\left(\alpha\right)=\left(U,\emptyset,\dr_U,U,U,\lambda_U\right)$.
    \item If $\alpha=\left(x,\varphi, y\right)$ is a path with $x\nearrow^Ay$ then its label is 
    \[
    \ev\left(\alpha\right)=\left(V,\emptyset,\dr_V,V\backslash A,V,\lambda_V\right)
    \]
    \item If $\alpha=\left(x,\varphi, y\right)$ is a path with $x\searrow_By$ then its label is 
    \[
    \ev\left(\alpha\right)=\left(U,\emptyset,\dr_U,U,U\backslash B,\lambda_U\right)
    \]
    \item If $\alpha=\beta_1*\beta_2*...*\beta_n$ the concatenation of steps $\beta_1,\beta_2,...,\beta_n$ then its label is the gluing composition of ipomsets $\ev\left(\alpha\right)=\ev\left(\beta_1\right)*\ev\left(\beta_2\right)*...*\ev\left(\beta_n\right)$.
\end{itemize}
\end{definition}
The labels of paths of length 0 or 1 are trivially interval ipomsets. 
Since the labelling of paths of length greater than 1 is defined as the gluing of the labels of its steps it follows that they are interval ipomsets as well.

\begin{example}
  Let us present the dashed paths in \cref{fig:transitions-hda} as a path $\alpha$ as follows.
  We will refer to the shaded $2$-cells as $x$ and $y$, respectively.
  The bottom-left corner of $x$ is $s$ and the top-right corner of $y$ is $e$.
  Moreover, we enumerate the $1$-cells labelled with $a$ and $c$ from left to right as
  $a^{k}$ and $c^{k}$ with $k \in \set{1,2,3}$.
  The path
  $\alpha = s \nearrow^{\set{0,1}} x \searrow^{\set{0}} a^{2} \nearrow^{\set{0}} y
  \searrow^{\set{0,1}} \delta^{^{1}}_{\set{0,1}}(y) \nearrow^{\set{0}} c^{3} \searrow^{\set{0}} e$
  enters from $s$ into $x$,
  crosses the $1$-cell $a^{2}$, the $2$-cell $y$ and the top-right corner $e$ of $y$, and finally
  goes through $c^{3}$ into the endpoint $e$ of $c^{3}$.
  Its label is given by
  \begin{equation*}
    \ev(\alpha) =
    \begin{pmatrix}
      a \bullet \\
      b \bullet
    \end{pmatrix}
    *
    \begin{pmatrix}
      \bullet a \bullet \\
      \bullet b
    \end{pmatrix}
    *
    \begin{pmatrix}
      d \bullet \\
      \bullet a \bullet
    \end{pmatrix}
    *
    \begin{pmatrix}
      \bullet d \\
      \bullet a
    \end{pmatrix}
    *
    \begin{pmatrix}
      c
    \end{pmatrix}
    * \varepsilon
    =
    \begin{tikzpicture}[pomset]
      \matrix[name=m, commutative diagrams/cramped] {
        & a &   \\
        &   & c \\
        b & d &   \\
      };
      \path[commutative diagrams/.cd, every arrow, every label]
      (m-1-2) edge (m-2-3)
      (m-3-1) edge (m-3-2)
      (m-3-2) edge (m-2-3);
    \end{tikzpicture}
  \end{equation*}
\end{example}

For a precubical set or HDA $X$ we define $P_X$ as the set of paths in $X$. For a path $\alpha=\left(x_0,\varphi_1,x_1,...,\varphi_n,x_n\right)$ we call $\pSrc\left(\alpha\right)=x_0$ the source and $\pTrg\left(\alpha\right)=x_n$ the target of the path. We can now define the languages of HDA.
\begin{definition} \label{hda-lang-def}
  The language of an HDA $X$ is the set of interval ipomsets
\[
L(X)=\left\{\ev\left(\alpha\right)\mid \alpha\in P_X\text{, }\pSrc\left(\alpha\right)\in X_{\bot}\text{, } \pTrg\left(\alpha\right)\in X^{\top}\right\}
\]
\end{definition}

We refer to a path $\alpha$ with $\pSrc\left(\alpha\right)\in X_{\bot}$ and $\pTrg\left(\alpha\right)\in X^{\top}$ as an accepting path. In \cref{hda-lang-interval} we will prove that for each HDA $X$ the language $L(X)$ of $X$ is a down-closed interval ipomset language, see \cref{lang-definition}.
Let $X$ and $Y$ be precubical sets with the precubical map $f:X\to Y$. For each path $\alpha=\left(x_0,\varphi_1,x_1,...,\varphi_n,x_n\right)$ in $X$ with $x_k\in X\left[U_k\right]$ we define
\begin{equation*}  f\left(\alpha\right)=\left(f_{U_0}\left(x_0\right),\varphi_1,f_{U_1}\left(x_1\right),...,\varphi_n,f_{U_n}\left(x_n\right)\right) \, ,
\end{equation*}
which by definition of the precubical maps is a path in $Y$.
The following to two lemmas show that precubical maps and HDA maps preserve paths and languages.
\begin{lemma} \label{path-theorem}
Let $X$ and $Y$ be precubical sets and let $f:X\to Y$ be a precubical map. Suppose that we have $\alpha, \beta\in P_X$ with $\pSrc\left(\alpha\right) = \pTrg\left(\beta\right)$. Then we have $\ev\left(\alpha*\beta\right)=\ev\left(\alpha\right)*\ev\left(\beta\right)$ and $\ev\left(f\left(\alpha\right)\right)=\ev\left(\alpha\right)$.
\end{lemma}
\begin{proof}
This follows directly from the definition of $\ev$.
\end{proof}

\begin{lemma} \label{language-subset}
Let $X$ and $Y$ be HDA and let $f:X\to Y$ be a HDA map. Then we have $L(X)\subseteq L(Y)$. If $f$ is an isomorphism then we have $L(X)=L(Y)$.
\end{lemma}
\begin{proof}
  If $P\in L(X)$ then there exists a path $\alpha$ in $X$ with $\pSrc\left(\alpha\right)\in \sCells{X}$ and $\pTrg\left(\alpha\right)\in \aCells{X}$ such that $\ev\left(\alpha\right)=P$.
  \Cref{path-theorem} gives us that $f\left(\alpha\right)$ is a path in $Y$ and because HDA maps preserve starting and accepting cells we have $\pSrc\left(f\left(\alpha\right)\right)\in \sCells{X}$ and $\pTrg\left(f\left(\alpha\right)\right)\in \aCells{X}$ and therefore $P=\ev\left(\alpha\right)=\ev\left(f\left(\alpha\right)\right)\in L(Y)$.

In the case that $f:X\to Y$ is an isomorphism there exists an inverse map $f^{-1}:Y\to X$, which gives us $L(Y)\subseteq L(X)$ as well and therefore $L(X)=L(Y)$.
\end{proof}

\subsection{Composition of HDA and their languages}
\label{sec:comp-hda}

In this section, we show how the colimits of languages and the languages of colimits of diagrams of HDA relate.
\begin{lemma} \label{colimit-languages}
Let $(X,\phi)$ be a cocone of the small diagram $F:D\to\HDA$. Then we have $\bigcup_{d\in D}L\left(F(d)\right)\subseteq L(X)$.
\end{lemma}
\begin{proof}
For every $d\in D$ we have the HDA map $\phi(d):F(d)\to X$. \Cref{language-subset} then gives us that $L\left(F(d)\right)\subseteq L(X)$, from which the statement follows.
\end{proof}

We get equality in the case that $(X,\phi)$ is a coproduct or a filtered colimit, as we will prove with in the next two results.
\begin{lemma} \label{coproduct-language}
  Let $F \from D\to\HDA$ be a small discrete diagram of HDA with
  coproduct $\left(X,\phi\right) = \colim F$.
  Then we have $\bigcup_{d\in D}L\left(F(d)\right)=L(X)$.
\end{lemma}
\begin{proof}[Proof of \cref{coproduct-language} on \cpageref{coproduct-language}]
Suppose that we have $P\in L(X)$. Then there exists an accepting path $\alpha=\left(x_0,\varphi_1,x_1,...,\varphi_n,x_n\right)$ in $X$ with $\pSrc(\alpha)\in \sCells{X}$ and $\pTrg(\alpha)\in \aCells{X}$ such that $\ev(\alpha)=P$.

Finite presentability gives us that for each $x_k\in X\left[U_k\right]$ for $1\leq k\leq n$ and the object $U_k\in\Cube$ there exists a unique $d_k\in D$ and a unique $y_k\in F(d)\left[U_k\right]$ such that $\phi_{d_k}\left[U_k\right]\left(y_k\right)=x_k$. It also gives us that $y_1\in\sCells{F\left(d_1\right)}$ and $y_n\in\aCells{F\left(d_n\right)}$.

Suppose that we have $x_k=\delta^0_A\left(x_{k+1}\right)$. Because we have
\[
\phi_{d_k}\left[U_k\right]\left(y_k\right)=x_k=\delta_A^0\left(x_{k+1}\right)=\delta^0_A\circ \phi_{d_{k+1}}\left[U_{k+1}\right]\left(y_{k+1}\right)=\phi_{d_k}\left[U_k\right]\circ\delta^0_A\left(y_{k+1}\right)
\]
we get $d_k=d_{k+1}$ and $y_k=\delta^0_A\left(y_{k+1}\right)$ from finite presentability.
Analogously the same works for if we have $\delta_B^1\left(x_k\right)=x_{k+1}$.

Therefore there exists an accepting path $\alpha'=\left(y_0,\varphi_1,y_1,...,\varphi_n,y_n\right)$ in $F(d)$ with $d=d_1=d_2=...=d_n$ such that $\phi_d\left(\alpha'\right)=\alpha$. \Cref{path-theorem} gives us that $P=\ev(\alpha)=\ev\left(\alpha'\right)$ and therefore $\ev\left(\alpha'\right)\in L\left(F(d)\right)$. As a result we have that $P\in L(X)\implies P\in\bigcup_{d\in D}L\left(F(d)\right)$. Combined with \cref{colimit-languages} this proves the statement.
\end{proof}

\begin{theorem} \label{hda-lang-filtered-colim}
  Let $F \from D \to \HDA$ be a small filtered diagram of HDA with
  filtered colimit $\left(X,\phi\right) = \colim F$.
  Then we have $\bigcup_{d\in D}L\left(F(d)\right)=L(X)$.
\end{theorem}
\Cref{hda-lang-filtered-colim} together with \cref{pcs-filtered-colimit-finite-pcs} shows that all HDA and their languages can be expressed as combination of finite HDA and union of languages. This powerful tool allows us to prove statements about the languages of HDA in a simple way by using the filtered colimits of finite HDA demonstrated by the following \namecref{hda-lang-interval}.

\begin{lemma} \label{hda-lang-interval}
The languages of HDA are down-closed interval ipomset languages.
\end{lemma}
\begin{proof}
For a finite HDA $X$, $L(X)$ is a language by \cite[Prop.~10]{FJS+22:KleeneTheoremHigherDimensional}.
Suppose that $X$ is an HDA.
From \cref{pcs-filtered-colimit-finite-pcs} we get a filtered diagram $F \from D\to\HDA$ of finite HDA such that $X\cong\colim_{d\in D}F(d)$. \cref{language-subset} and \cref{hda-lang-filtered-colim} give us that
\[
L(X) = L \parens[\Big]{\colim_{d\in D}F(d)} = \bigcup\nolimits_{d\in D}L\left(F(d)\right)
\]
The result follows because languages are closed arbitrary unions, see \cref{down-closure-union}.
\end{proof}

Since $\Lang$ is the category having down-closed interval ipomset languages as objects and the subset inclusion maps as morphisms, \cref{hda-lang-interval,language-subset} allow us to see $L$ as a functor $L:\HDA\to\Lang$. Since the colimit of a diagram of languages is the union, \cref{coproduct-language,hda-lang-filtered-colim} give us that $L$ preserves coproducts and filtered colimits. However, it does not preserve all colimits as we show with the next \namecref{hda-lang-colim-counter}.
\begin{proposition} \label{hda-lang-colim-counter}
  There is a small diagram $F \from D \to \HDA$, whose colimit accepts more than the HDA in the diagram together:
  $\bigcup_{d\in D}L\left(F(d)\right) \subsetneq L(\colim  F)$.
\end{proposition}
\begin{proof}
  We use for $D$ the category of shape $1 \leftarrow 2 \rightarrow 3$.
  Consider the following pushout of HDA, which is a colimit over a diagram of shape $D$.
  \begin{equation*}
\begin{tikzcd}
	{(\circ)} & {(\Rightarrow \bullet \xrightarrow{a} \circ)} \\
	{(\circ \xrightarrow{c} \bullet \Rightarrow)} & {(\Rightarrow \bullet \xrightarrow{a} \bullet \xrightarrow{c} \bullet \Rightarrow)}
	\arrow["{i_2}"', from=1-1, to=2-1]
	\arrow["{i_1}", from=1-1, to=1-2]
	\arrow[from=1-2, to=2-2]
	\arrow[from=2-1, to=2-2]
	\arrow["\lrcorner"{anchor=center, pos=0.125, rotate=180}, draw=none, from=2-2, to=1-1]
\end{tikzcd}
  \end{equation*}
  The inclusions $i_{k}$ map $\circ$ to $\circ$ and the double arrows indicate starting
  and accepting cells.
  Note that the languages of the HDA at the corners are all empty, except of the HDA at the
  bottom right corner, which accepts the word $(a \to c)$.
  Thus the pushout of these HDA with empty languages has a non-empty language.
  \qed
\end{proof}
Finally, we prove that the language of the tensor product of two HDA is the same as the parallel composition of their two individual languages.
\begin{theorem} \label{languages-tensor}
  The functor $L \from (\HDA, \otimes, I) \to (\Lang, \parallel, \set{\emptyLO}) $ is strict monoidal.
\end{theorem}
\begin{proof}
  Let $X$ and $Y$ be HDA.
  We have to show that $L\left(X\otimes Y\right)=L\left(X\right)\parallel L\left(Y\right)$.
  \Cref{pcs-filtered-colimit-finite-pcs} provides use with filtered diagrams
$F \from D\to\HDA$ and $G \from E\to\HDA$ of finite HDA with $X$ and $Y$ being their respective
filtered colimits.
This allows us to generalise \cite[Prop.~19]{FJS+22:KleeneTheoremHigherDimensional}, where
$L(X \otimes Y) = L(X) \parallel L(Y)$ is proved for finite HDA, to arbitrary HDA.
\begin{align*}
  L\left(X\otimes Y\right)
  & = L\parens[\Big]{\colim_{(d,e)\in D\times E}F(d)\otimes G(e)}
    \tag*{tensor product preserves colimits} \\
  & = \bigcup\nolimits_{(d,e)\in D\times E}L\left(F(d)\otimes G(e)\right)
    \tag*{by \cref{hda-lang-filtered-colim}} \\
  & = \bigcup\nolimits_{(d,e)\in D\times E}L\left(F(d)\right)\parallel L\left(G(e)\right)
  \tag*{\cite[Prop.~19]{FJS+22:KleeneTheoremHigherDimensional} for finite HDA} \\
  & = \bigcup\nolimits_{d\in D}L\left(F(d)\right) \, \parallel \, \bigcup\nolimits_{e\in E}L\left(G(e)\right)
    \tag*{by \cref{parallel-lang-union-commute}} \\
  & = L(X)\parallel L(Y)
    \tag*{by \cref{hda-lang-filtered-colim}}
\end{align*}
This shows that even for arbitrary HDA the parallel composition of their languages
is given by tensoring the HDA.
That $L(I) = \set{\emptyLO}$ is obvious.
\qed
\end{proof}

\section{Process Replication as Rational HDA}
\label{sec:replication}

In this section, we seek to complete the correspondence between concurrent Kleene algebras and HDA,
which requires us to identify a notion of \emph{rational HDA} that can capture finitely presented behaviour.
This has almost been accomplished~\cite{FJS+22:KleeneTheoremHigherDimensional} but the parallel closure
could not be realised as finite HDA.
For regular languages, linear weighted languages and various other languages without true
concurrency, the correspondence between languages and automata has been studied from a coalgebraic
perspective~\cite{BMS13:SoundCompleteAxiomatizations,Milius10:SoundCompleteCalculus,%
  Milius13:RationalOperationalModels}.
We make in \cref{sec:locally-compact} a first attempt, where we follow these ideas by studying
locally compact HDA and by showing how to realise the parallel closure as locally compact HDA.
However, we will see that this model is too powerful and will restrict to finitely branching
HDA in \cref{sec:finitely-branching}.
These can realise the parallel Kleene star as well, but will require an infinite choice
at the start.
Thus, none of these choices is satisfactory to act as rational HDA and we show that it is
impossible to realise the parallel closure as finitely branching HDA with finitely many
starting cells.

\subsection{Locally Compact HDA}
\label{sec:locally-compact}

Let us first define what we mean by locally compact HDA.
This follows work on rational coalgebraic
behaviour~\cite{Milius13:RationalOperationalModels,Milius10:SoundCompleteCalculus} and
can be seen as axiomatisation of the factorisation property that filtered colimits enjoy
in lfp categories.

\begin{definition}
  \label{def:locally-compact-hda}
  A HDA $\HDATup{X}$ is \emph{locally compact} if the forgetful functor
  $\forget \from \Comma{\HDAcomp}{X}\rightarrow \Comma{\presheaf{\Cube}_c}{X}$ is
  cofinal.
  Explicitly, this means~\cite[0.11]{adamek_rosicky_1994} that
  1) for all compact precubical set $P$ and $f \from P \to X$
  there is a factorisation of $f$ into $P \xrightarrow{f'} Y \xrightarrow{h} X$,
  where $h \from \HDATup{Y} \to \HDATup{X}$ is a
  HDA morphism and $\HDATup{Y}\in \HDAcomp$; and
  2) for all other $\HDATup{Y'} \in \HDAcomp$, $h' \from \HDATup{Y'} \to \HDATup{X}$
  and $f'' \from P \to Y'$ with
  $h' \comp f'' = f$, there exist $\HDATup{R}\in \HDAcomp$ and HDA morphisms
  $e \from \HDATup{Y'} \to \HDATup{R}$ and $e' \from \HDATup{Y} \to \HDATup{R}$ such that
  $e' \comp f' = e \comp f''$, see the following diagrams.
  \begin{equation*}
    \begin{tikzcd}
      P \ar[r, "f"] \ar[dr, pos=0.4,"f'"{below=2pt}]
      & X \\
      & Y \ar[u, "h"{right}]
    \end{tikzcd}
    \qquad
    \qquad
\begin{tikzcd}
	P & Y \\
	{Y'} & R
	\arrow["{f'}", from=1-1, to=1-2]
	\arrow["{f''}"', from=1-1, to=2-1]
	\arrow["{e'}", from=1-2, to=2-2]
	\arrow["e"', from=2-1, to=2-2]
\end{tikzcd}
\end{equation*}
\end{definition}

The following lemma shows that the second condition is redundant, which follows from $\HDA$ being
an lfp category.
\begin{lemma}
  \label{essentially-unique-redundant}
  An HDA $\HDATup{X}$ is locally compact if and only if all presheaf morphisms
  $f \from P \to X$ factor as
  $P \xrightarrow{f'} Y \xrightarrow{h} X$ into a presheaf morphism $f'$ and
  a HDA morphism $h \from \HDATup{Y} \to \HDATup{X}$ from $\HDATup{Y}\in \HDAcomp$.
\end{lemma}
\begin{proof}
  Suppose that we are given compact HDA
  $\mathcal{Y} = \HDATup{Y}$ and
  $\mathcal{Y}' = \HDATup{Y'}$ with morphisms
  $h \from \mathcal{Y} \to \HDATup{X}$,
  $h' \from \mathcal{Y'} \to \HDATup{X}$,
  $g \from P \to Y$
  and $g' \from P \to Y'$, such that
  $h \comp g = f$ and
  $h' \comp g' = f$.
  Since $\HDAcomp$ is closed under finite colimits, we can form the coproduct $\mathcal{Y} + \mathcal{Y}'$
  with inclusions $\kappa$ and $\kappa'$.
  Let $\copair{h,h'} \from \mathcal{Y} + \mathcal{Y}' \to \mathcal{X}$ be the copairing of $h$ and $h'$, where
  $\mathcal{X} = \HDATup{X}$.
  Because $\HDA$ is lfp, we can factor $\copair{h,h'}$ into an epimorphism $q$ and a monomorphism $m$:
  $\copair{h,h'} = \mathcal{Y} + \mathcal{Y}' \xrightarrow{e} \mathcal{R} \xrightarrow{m} \mathcal{X}$.
  We then define $e = q \comp \kappa$ and $e' = q \comp \kappa'$.
  Note that because $q$ is an epimorphism, $\mathcal{R}$ is a compact HDA.
  With this notation set up, we have
  \begin{equation*}
    me'g' = mq \kappa' g'= \copair{h,h'}\kappa' g' = h'g' = hg = \copair{h,h'} \kappa g = mq\kappa g = meg
  \end{equation*}
  and thus, since $m$ is mono, $e'g' = eg$.
\end{proof}

The next theorem shows that local compactness can be derived from the presentation of HDA as
filtered colimit of compact HDA.
\begin{theorem}
  \label{thm:all-hda-locally-compact}
  All HDA are locally compact.
\end{theorem}
\begin{proof}
  Let $X$ be an HDA.
  By \cref{pcs-filtered-colimit-finite-pcs}, $X$ is isomorphic to the
  colimit $\colim U_{X}$ of a filtered diagram $U_{X} \from D \to \HDAcomp$.
  Therefore, $\colim(D \to \HDAcomp \to \HDA)$ is locally compact because filtered colimits in
  lfp categories factor essentially uniquely through colimit inclusions.

  If $X$ is locally compact, then every $x \in X[U]$ generates a compact
  sub-precubical set $\pair{x} \hookrightarrow X$ that contains $x$ and all its boundary cells.
  This inclusion factors essentially uniquely into an inclusion of a compact HDA into
  $\colim U_{X}$ for every $U$ and $x \in X[U]$ by local compactness.
  It is easy to see that these inclusions jointly set up an isomorphism.
\end{proof}

\Cref{thm:all-hda-locally-compact} shows that local compactness is no restriction in the case of HDA,
contrary to other computational models.
Let us, nevertheless, apply the lessons of local compactness to get closer to an HDA that models
process replication in a reasonably finitary way.
Before that, let us warm up and construct a HDA as a filtered colimit with infinite branching.
\begin{example}
  \label{ex:infinite-branching}
  Let $F \from \Cat{D} \to \HDAcomp$ be the diagram given by
  \begin{equation*}
    \begin{tikzcd}[anchor=south]
      0 \ar[r, "a"] & 1
    \end{tikzcd}
    \quad \longrightarrow \quad
    \begin{tikzcd}[row sep=small, baseline=4pt,anchor=south]
      & 2 \\
      0 \ar[r, "a"] \ar[ur, "a"]
      & 1
    \end{tikzcd}
    \quad \longrightarrow \quad
    \begin{tikzcd}[row sep=tiny, baseline=4pt,anchor=south]
      3 & \\[-7pt]
      & 2 \\
      0 \ar[r, "a"]  \ar[ur, "a"] \ar[uu, "a"]
      & 1
    \end{tikzcd}
    \quad \longrightarrow \quad \dotsm
  \end{equation*}
  This is a chain and thus filtered, and its colimit a $1$-dimensional HDA with infinitely many
  branches coming out of $0$.
  Nevertheless, since each HDA in the chain is compact, $\colim F$ is locally compact.
\end{example}

\begin{example}
  \label{ex:locally-compact-replication}
  Similarly to \cref{ex:infinite-branching}, we can also branch with higher dimensions
  and thus realise process replication as filtered colimit of compact HDA.
  For the purpose of this example it is simpler to ignore starting cells, but it is
  easy to see that tensor product and colimits are not affected by this.

  Let $A$ be the HDA with one $1$-cell labelled with $a$ and the endpoint of this $1$-cell
  taken as accepting.
  This is illustrated in \cref{fig:process-replication-construction} on the left, where the double
  arrows mark accepting cells.
  \begin{figure}[ht]
    \vspace*{-2em}
    \begin{equation*}
      \begin{tikzcd}
        {} \\
        & \bullet \\
        & 0
        \arrow["a", from=3-2, to=2-2]
        \arrow[shorten >=22pt, Rightarrow, from=2-2, to=1-1]
      \end{tikzcd}
      \qquad
      \xrightarrow{d_{1}}
      \begin{tikzcd}
        {} & {} \\
        & \bullet & \bullet \\
        & 0 & \bullet
        \arrow["a"', from=3-2, to=3-3]
        \arrow["a", from=3-2, to=2-2]
        \arrow["a"', from=3-3, to=2-3]
        \arrow["a", from=2-2, to=2-3]
        \arrow[shorten >=22pt, Rightarrow, from=2-3, to=1-2]
        \arrow[shorten >=22pt, Rightarrow, from=2-2, to=1-1]
      \end{tikzcd}
      \quad
      \xrightarrow{d_{2}}
      \begin{tikzcd}[sep=small]
        {} && \bullet && \bullet & {} \\
        & \bullet && \bullet && {} \\
        && \bullet && \bullet \\
        & \bullet && \bullet
        \arrow["a"{description}, from=1-3, to=1-5]
        \arrow[shorten >=3pt, Rightarrow, from=1-5, to=1-6]
        \arrow[shorten >=6pt, Rightarrow, from=2-2, to=1-1]
        \arrow["a"{description}, from=2-2, to=1-3]
        \arrow["a"{description, pos=0.3}, from=2-2, to=2-4]
        \arrow["a"{description}, from=2-4, to=1-5]
        \arrow[shorten >=24pt, Rightarrow, from=2-4, to=2-6]
        \arrow["a"{description, pos=0.3}, from=3-3, to=1-3]
        \arrow["a"{description, pos=0.3}, from=3-3, to=3-5]
        \arrow["a"{description}, from=3-5, to=1-5]
        \arrow["a"{description}, from=4-2, to=2-2]
        \arrow[from=4-2, to=3-3]
        \arrow["a"{description}, from=4-2, to=4-4]
        \arrow["a"{description, pos=0.7}, from=4-4, to=2-4]
        \arrow["a"{description}, from=4-4, to=3-5]
      \end{tikzcd}
      \xrightarrow{d_{3}} \dotsm
    \end{equation*}
    \vspace*{-0.5cm}
    \caption{Chain of HDA to construct process replication of HDA $A$ on the left.
      The starting cell named $0$ shows how $d_{n}$ embeds the cells matching
      with the accepting cells.}
    \label{fig:process-replication-construction}
  \end{figure}
  The maps $d_{n} \from A_{n} \to A_{n+1}$ in \cref{fig:process-replication-construction},
  where $A_{1} = A$, are constructed as in the following pushout diagram.
  In this diagram, we denote by $A^{\otimes n}$ the $n$-fold tensor product of $A$ with itself,
  where $A^{\otimes 0} = I$.
  For an HDA $X$, we write $X^{\varepsilon}$ for the HDA that has the same underlying precubical
  set but no starting and accepting states.
  \begin{equation*}
    \begin{tikzcd}[row sep=small]
      {A^{\otimes n,\varepsilon}} & {A^{\otimes n,\varepsilon} \otimes I} & {A^{\otimes n+1}} & {A^{\otimes n+1,\varepsilon}} \\
      \\
      {A_n} && {A_{n+1}}
      \arrow["\cong", from=1-1, to=1-2]
      \arrow[from=1-2, to=1-3]
      \arrow["{i_{n}}"', from=1-1, to=3-1]
      \arrow["{d_{n}}"', from=3-1, to=3-3]
      \arrow[from=1-3, to=3-3]
      \arrow["\lrcorner"{anchor=center, pos=0.125, rotate=180}, draw=none, from=3-3, to=1-1]
      \arrow[from=1-4, to=1-3]
      \arrow["{i_{n+1}}"{description}, from=1-4, to=3-3]
    \end{tikzcd}
  \end{equation*}
  Intuitively, the HDA $A_{n+1}$ is given by extending $A_{n}$ to a full $n+1$-dimensional cube,
  where $A_{n}$ is included via $d_{n}$ as the ``front face''.
  In \cref{fig:process-replication-construction}, this inclusion is indicated by the vertex $0$,
  which is identified via $d_{n}$.
  The indicated maps $d_{n}$ form a chain of compact HDA and thus a filtered diagram
  $F \from \Cat{D} \to \HDAcomp$.
  By taking the colimit of $F$ and declaring the cell marked $0$ as starting cell, we
  obtain an HDA that accepts $L(A)^{(\ast)}$, the parallel Kleene closure of the language of
  $A$.
  That this is the case follows directly from \cref{languages-tensor} and
  \cref{hda-lang-filtered-colim}.
  Since each HDA in the chain is compact, $\colim F$ is locally compact, but this colimit is a
  HDA with infinitely many branches coming out of $0$.
\end{example}

\subsection{Finitely Branching HDA}
\label{sec:finitely-branching}

The HDA that we constructed in \cref{ex:locally-compact-replication} has the pleasant property
that during execution many $a$-processes can be spawned, as one would expect from a process
replication operator that occurs in process algebra.
However, the HDA in \cref{ex:locally-compact-replication} has infinitely many cells branching out
of any cell.
This makes it impossible to realise this HDA on a physical machine and motivates another possible
definition of what one may consider rational HDAs.

\begin{definition}
  \label{def:finitely-branching-hda}
  A HDA $X$ is \emph{finitely branching} if for all lo-sets $U \cup \set{a}$ and all $x \in X_{U}$ the
  set $\setDef{y \in X_{U \cup \set{a}}}{\delta_{A,B}(y) = x}$ is finite.
  We denote by $\HDAfb$ the full subcategory of $\HDA$ that consists of
  finitely branching HDA.
\end{definition}

Clearly, finitely branching HDA are not closed under filtered colimits, as
\cref{ex:locally-compact-replication} shows.
However, they are closed under coproducts.

\begin{lemma}
  \label{fb-hda-coproduct-closed}
  Let $F \from \Cat{D} \to \HDAfb$ a diagram on a small discrete category $D$.
  Then the colimit (coproduct) $\colim F$ exists in $\HDAfb$.
\end{lemma}

The parallel Kleene star of a finitely branching HDA $X$, also known as process replication, to obtain the parallel Kleene closure of its language, see \cref{def:paralell-lang}, can be
realised as finitely branching HDA.
We write $X^{\otimes n}$ for the $n$-fold tensor product of $X$ with itself, where
$X^{\otimes 0} = I$, and define the parallel replication of $X$ to be
$!X = \coprod_{n \in \N}X^{\otimes n}$.

\begin{theorem}
  \label{thm:process-replication-finite-branching}
  The HDA $!X$ is finitely branching and we have $L(!X) = L(X)^{(\ast)}$.
\end{theorem}
\begin{proof}
  By \cref{coproduct-language} and \cref{languages-tensor} we have
  \begin{equation*}
    L(!X)
    = L\parens*{\coprod\nolimits_{n\in\N}X^{\otimes n}}
    = \bigcup\nolimits_{n \in \N}L\parens*{X^{\otimes n}}
    = \bigcup\nolimits_{n \in \N}L\parens*{X}^{\parallel n}
    = L(X)^{(\ast)}
  \end{equation*}
  That $!X$ is finitely branching is given by \cref{fb-hda-coproduct-closed}.
\end{proof}

The caveat of this theorem, and the definition of finitely branching in general, is that
we do not make any restrictions on the number of starting cells.
In fact, $!X$ will have infinitely many starting cells, if $X$ has at least one.

\begin{example}
  \label{ex:finitely-branching-replication}
  Let $A$ again be the HDA as in \cref{ex:locally-compact-replication}.
  The HDA $!A$ looks as in \cref{fig:process-replication-finitely-branching}.
  \begin{figure}[ht]
    \begin{equation*}
      \begin{tikzcd}[column sep=small]
        & {} \\
        1 & \bullet
        \arrow["a", from=2-1, to=2-2]
        \arrow[shorten >=5pt, Rightarrow, from=2-2, to=1-2]
      \end{tikzcd}
      \qquad
      \begin{tikzcd}
        \bullet & \bullet & {} \\
        2 & \bullet
        \arrow["a"{description}, from=1-1, to=1-2]
        \arrow[shorten >=8pt, Rightarrow, from=1-2, to=1-3]
        \arrow["a"{description}, from=2-1, to=1-1]
        \arrow["a"{description}, from=2-1, to=2-2]
        \arrow["a"{description}, from=2-2, to=1-2]
      \end{tikzcd}
      \quad
      \begin{tikzcd}[sep=small]
        & \bullet && \bullet & {} \\
        \bullet && \bullet \\
        & \bullet && \bullet \\
        3 && \bullet
        \arrow["a"{description}, from=1-2, to=1-4]
        \arrow[shorten >=5pt, Rightarrow, from=1-4, to=1-5]
        \arrow["a"{description}, from=2-1, to=1-2]
        \arrow["a"{description, pos=0.3}, from=2-1, to=2-3]
        \arrow["a"{description}, from=2-3, to=1-4]
        \arrow["a"{description, pos=0.3}, from=3-2, to=1-2]
        \arrow["a"{description, pos=0.3}, from=3-2, to=3-4]
        \arrow["a"{description}, from=3-4, to=1-4]
        \arrow["a"{description}, from=4-1, to=2-1]
        \arrow[from=4-1, to=3-2]
        \arrow["a"{description}, from=4-1, to=4-3]
        \arrow["a"{description, pos=0.7}, from=4-3, to=2-3]
        \arrow["a"{description}, from=4-3, to=3-4]
      \end{tikzcd}
      \dotsm
    \end{equation*}
    \vspace*{-0.3cm}
    \caption{Finitely branching HDA for replication of $A$ constructed as coproduct, where
      the cells $1,2,3,\dotsc$ are starting cells and double arrows mark
      accepting cells}
    \label{fig:process-replication-finitely-branching}
  \end{figure}
  Notice that it consists of little finite islands, each with a starting cell.
  The HDA has to make at the beginning of an execution a choice on the
  number of parallel executions of the action $a$.
  This means that this HDA is not realisable, as such a guess requires knowledge about
  how many parallel processes will be needed.
  For instance, a web server would need to know \emph{when it is started} how many clients
  will connect during its life time.
  This is clearly impossible.
\end{example}

The \cref{ex:locally-compact-replication,ex:finitely-branching-replication} show that
either way of realising process replication, as locally compact HDA or as finitely branching
HDA, leads to operational problems.
In fact, it is not possible to realise process replication as finitely branching HDA with
finite starting cells.

\begin{theorem}
  \label{thm:kleene-par-star-impossible}
  There is no HDA $X \in \HDAfb$ with finitely many initial states, such that $X$ would realise
  the parallel Kleene star of $L(A) = \set{(a)}$, where $A$ is the HDA with only one $a$-transition,
  as in \cref{ex:locally-compact-replication}.
\end{theorem}
\begin{proof}
  Suppose there is an HDA $X \in \HDAfb$ with finite initial states, such that
  $L(X) = L(A)^{(\ast)} = \set{(a)}^{(\ast)}$.
  We partition $L(X)$ into languages $L_{x}$ for $x \in \sCells{X}$.
  Since $\sCells{X}$ is finite, some $L_{x}$ must be infinite.
  Thus for every $\set{(a)}^{\parallel n} \in L_{x}$
  there must be an $n$-cell of which $x$ is a boundary.
  But then $X$ has infinitely many branches at $x$, and thus $X$ cannot exist
  with the proclaimed properties.
\end{proof}

Since the identity language for gluing has infinite width~\cite[Example~4]{FJS+22:KleeneTheoremHigherDimensional},
it cannot be presented by a finite HDA.
One can provide a finitely branching HDA that accepts the identity language, but with
infinitely many starting cells.
Thus, even this simple language does not fit into any reasonable restriction of HDA.

\section{Conclusion}
\label{sec:concl}

What does this leave us with?
The problem is that HDA combine state space and transitions into one object, a precubical set.
Intuitively, this prevents us from having transitions and cycles among cells of higher dimension.
More technically, the locally compact HDA allow infinite branching, while finite branching limits
the number of active parallel events to be finite.
This can be compared to the coalgebras for the finite powerset functor, also known as finitely
branching transition systems.
Here, locally compact transition systems may only have finite branching and thus realise
locally the behaviour of finite transition systems, as one would expect.
Therefore, one is led to the conclusion that, despite their semantic value, HDA as an operational
computational model are unsuited to
model process replication and another model for true concurrency has to be sought.
This is not to say that topological or geometrical models, like HDA, are inherently flawed but rather
that they have to be expanded to allow for the dynamic spawning of processes, in contrast to the
static nature of HDA.

\emergencystretch=1em
\bibliographystyle{alpha}
\bibliography{hda-colimits.bib}

\appendix

\section{Notation}
\label{sec:notation}

\begin{tabular}{l|Hl}
  Notation & Macro & Meaning \\\hline
  $\StdCat{C}$ & \macroA{StdCat}{C} & Standard or specific categories \\
  $\SetC$ & \macro{SetC} & Category of sets \\
  $\TopC$ & \macro{TopC} & Category of topological spaces \\
  $\Yo$ & \macro{Yo} & Yoneda embedding \\
  $\Sigma$ & \macro{Sigma} & Fixed alphabet \\
  $\car{P}$ & \macroA{car}{P} & Carrier of iposet $P$ \\
  $\downCl{A}$ & \macroA{downCl}{A} & Downwards closure \\
  $\varepsilon$ & \macro{emptyLO} & empty lo-set \\
  $\lSLO$ & \macro{lSLO} & category of labelled strict linear orders \\
  $\sloTens$ & \macro{sloTens} & monoidal product of $\lSLO$ \\
  $\fOrd{n}$ & \macroA{fOrd}{n} & finite ordinal with $n$ elements (possibly empty!) \\
  $\spine{n}$ & \macroA{spine}{n} & finite ordinal with $n + 1$ elements (spine of $n$-simplex) \\
  $\FCube$ & \macro{FCube} & Full labelled precube category \\
  $\Cube$ & \macro{Cube} & Labelled precube category (skeletal) \\
  $d_{A,B}$ & \macro{d\_\{A,B\}} & Coface map arising from the inclusion
                                   $U \setminus (A \cup B) \to U$ \\
  $\HDA$ & \macro{HDA} & Category of HDA \\
  $\Cat{C}$ & \macroA{Cat}{C} & Generic category \\
  $\op{\Cat{C}}$ & \macroA{op}{\macroA{Cat}{C}} & Opposite category \\
  $\presheaf{\Cat{I}}$ & \macroA{presheaf}{\macroA{Cat}{I}} & $\SetC$-Valued presheaves indexed by $\Cat{I}$ \\
  $\sCells{X}$ & \macroA{sCells}{X} & Starting cells of HDA \\
  $\aCells{X}$ & \macroA{aCells}{X} & Accepting cells of HDA \\
  $\HDATup{X}$ & \macroA{HDATup}{X} & Tuple that makes an HDA\\
  $\Lang$ & \macro{Lang} & Category of languages\\
  $\iiPom$ & \macro{iiPoms} & The set of interval ipomsets \\
  $\pSrc(\alpha)$ & \macro{pSrc} & Source of path $\alpha$ in an HDA \\
  $\pTrg(\alpha)$ & \macro{pSrc} & Source of path $\alpha$ in an HDA
\end{tabular}

\section{Convolution Product on HDA}
\label{sec:convolution}

\subsection{Day Convolution Precubical Sets is Coproduct}

In \cref{def:hda-tensor} we defined the tensor products of HDA as extending the tensor
product of precubical sets given by Day convolution with appropriate starting and accepting cells.
We show here that the coend formula
\begin{equation}
  \label{eq:hda-tensor-rep}
  X \otimes Y = \int^{V,W} \Cube(-, V \cubeTens W) \times X[V] \times Y[W]
\end{equation}
for Day convolution reduces to a coproduct formula
\begin{equation}
  \label{eq:hda-tensor-coproduct}
  (X \otimes Y)(U) \cong \coprod_{U = V \cubeTens W} X[V] \times Y[W]
\end{equation}
and thus reduces to the standard definition~\cite{BH87:TensorProductsHomotopies,%
  Grandis09:DirectedAlgebraicTopology,Kan55:AbstractHomotopy}

Recall that objects in $\lSLO$ are pairs $(\fOrd{n}, w)$ where $n \in \N$ and $w$ is
a word of length $n$ over $\Sigma$.
Let us write $i_{n,j} \from \fOrd{n} \to \fOrd{n+1}$ for the unique map that does not have $j$
in its image.
Clearly, any map $(\fOrd{n}, w) \to (\fOrd{n+1}, w')$ is determined by the embedding maps
$i_{n,j}$.
Therefore, we will leave out in the remainder the words $w$ and pretend that $\lSLO$ consists
of unlabelled finite ordinals $\fOrd{n}$.
Further, a map $d \from \fOrd{n} \to \fOrd{n+1}$ in $\Cube$ comes with a partition of the
complement image and is therefore given by either $(i_{n,j}, \set{j}, \emptyset)$
or $(i_{n,j}, \emptyset, \set{j})$.
For what follows, this duplication of morphisms also makes no difference and we focus attention
on the maps $i_{n,j}$.

The strategy to show that \cref{eq:hda-tensor-coproduct} holds is to show that any cowedge for
the coend in \cref{eq:hda-tensor-rep} is uniquely determined by a cocone for the coproduct in
\cref{eq:hda-tensor-coproduct}.
Write $F_{n,X,Y} \from \Cube \times \Cube \times \op{\Cube} \times \op{\Cube} \to \SetC$
for the functor given by
\begin{equation*}
  F_{n,X,Y}(\fOrd{m}, \fOrd{k}, \fOrd{m'}, \fOrd{k'})
  = \Cube(\fOrd{n}, \fOrd{m} \cubeTens \fOrd{k}) \times X_{m'} \times Y_{k'}
\end{equation*}
on objects, which gives us
$(X \otimes Y)_{n} = \int^{\fOrd{m},\fOrd{k}}F_{n,X,Y}(\fOrd{m},\fOrd{k},\fOrd{m},\fOrd{k})$.
Suppose now that $f \from F \to C$ is a cowedge, which means that it consists of maps
$f_{m,k} \from \Cube(\fOrd{n}, \fOrd{m} \cubeTens \fOrd{k}) \times X_{m} \times Y_{k} \to C$
in $\SetC$, such that the following diagram commutes for all $u \from \fOrd{m} \to \fOrd{m'}$
and $v \from \fOrd{k} \to \fOrd{k'}$.
\begin{equation*}
\begin{tikzcd}[column sep=2.25em]
	& {\Cube(\fOrd{n}, \fOrd{m'} \cubeTens \fOrd{k'}) \times X_{m'} \times Y_{k'}} \\
	{\Cube(\fOrd{n}, \fOrd{m} \cubeTens \fOrd{k}) \times X_{m'} \times Y_{k'}} && C \\
	& {\Cube(\fOrd{n}, \fOrd{m} \cubeTens \fOrd{k}) \times X_m \times Y_k}
	\arrow["{f_{m',k'}}", shorten >=4pt, from=1-2, to=2-3]
	\arrow["{\id \times X(u) \times Y(v)}"'{pos=0.4}, from=2-1, to=3-2]
	\arrow["{\Cube(\fOrd{n}, u \cubeTens v) \times \id \times \id}", from=2-1, to=1-2]
	\arrow["{f_{m,k}}"', shorten >=4pt, from=3-2, to=2-3]
\end{tikzcd}
\end{equation*}

Suppose now that $n = m + k$ and consider the following diagram, which commutes for all
appropriate choices of $j$ since $f$ is a cowedge.
\begin{equation*}
\begin{tikzcd}[column sep=small,row sep=large]
	& {\Cube(\fOrd{n}, (\fOrd{m+1}) \cubeTens (\fOrd{k-1})) \times X_{m+1} \times Y_{k-1}} \\
	{\Cube(\fOrd{n}, (\fOrd{m+1}) \cubeTens (\fOrd{k-1})) \times X_{m+1} \times Y_{k}} \\
	{\Cube(\fOrd{n}, (\fOrd{m+1}) \cubeTens \fOrd{k}) \times X_{m+1} \times Y_k} && C \\
	{\Cube(\fOrd{n}, \fOrd{m} \cubeTens \fOrd{k}) \times X_{m+1} \times Y_k} \\
	& {\Cube(\fOrd{n}, \fOrd{m} \cubeTens \fOrd{k}) \times X_m \times Y_k}
	\arrow["{f_{m+1,k}}", from=3-1, to=3-3]
	\arrow["{\id \times X(i_{m,j}) \times \id}"'{pos=0.4}, from=4-1, to=5-2]
	\arrow["{\Cube(\fOrd{n}, i_{m,j} \cubeTens \id) \times \id}"{description}, from=4-1, to=3-1]
	\arrow["{f_{m,k}}", curve={height=12pt}, from=5-2, to=3-3]
	\arrow["{\Cube(\fOrd{n}, \id \cubeTens i_{k-1,j}) \times \id}"{description}, from=2-1, to=3-1]
	\arrow["{\id \times \id \times Y(i_{k-1,j})}"{pos=0.4}, from=2-1, to=1-2]
	\arrow["{f_{m+1,k-1}}"', curve={height=-12pt}, from=1-2, to=3-3]
\end{tikzcd}
\end{equation*}

But then $f_{m+1,k}$ is determined from $f_{m+1,k-1}$ and $f_{m,k}$, since any map
$\fOrd{n} \to (\fOrd{m+1}) \cubeTens \fOrd{k}$ is uniquely determined by the only number $j$
that is not in its image.
These are exactly the maps obtained as the image of the maps
$\Cube(\fOrd{n}, i_{m,j} \cubeTens \id)$ and $\Cube(\fOrd{n}, \id \cubeTens i_{k-1,j})$.
Hence, the parts in the coend of \cref{eq:hda-tensor-rep} where $n < k + m$ do not contribute
and it suffices to consider splittings of $n = m + k$.
This gives us \cref{eq:hda-tensor-coproduct}.

\end{document}